\documentclass[11pt]{article}

	\usepackage{a4,geometry}

\usepackage{graphicx}
\usepackage{amsmath,amssymb,amsthm,mathtools}
\usepackage{paralist}
\usepackage{bm}
\usepackage{bbm}
\usepackage{xspace}
\usepackage{url}
\usepackage{fullpage, prettyref}
\usepackage{boxedminipage}
\usepackage{wrapfig}
\usepackage{ifthen}
\usepackage{color}
\usepackage{xcolor}
\usepackage{framed}
\usepackage{algorithmic,algorithm}
\usepackage{fullpage}
\usepackage{varioref}
\usepackage[pagebackref,colorlinks=true,pdfpagemode=none,urlcolor=blue,linkcolor=blue,citecolor=violet,pdfstartview=FitH]{hyperref}

\newtheorem{theorem}{Theorem}[section]
\newtheorem{conjecture}[theorem]{Conjecture}
\newtheorem{lemma}[theorem]{Lemma}
\newtheorem{claim}[theorem]{Claim}

\newtheorem{definition}[theorem]{Definition}

\newcommand{\ignore}[1]{}

\newcommand{\cP}{\mathcal{P}}

\newcommand{\cS}{\mathcal{S}}
\newcommand{\cT}{{\cal T}}

\newcommand{\eps}{\varepsilon}

\newcommand{\poly}{\mathrm{poly}}

\newcommand{\dirinf}[2]{\mathrm{Inf}^+_#1(#2)}
\newcommand{\dirinfc}[3]{\mathrm{Inf}^+_{#1,#2}(#3)}
\newcommand{\dirvol}[1]{\mathrm{\mu}^+(#1)}

\newcommand{\Sec}[1]{\hyperref[sec:#1]{\S\ref*{sec:#1}}} 
\newcommand{\Eqn}[1]{\hyperref[eq:#1]{(\ref*{eq:#1})}} 
\newcommand{\Fig}[1]{\hyperref[fig:#1]{Fig.\,\ref*{fig:#1}}} 
\newcommand{\Tab}[1]{\hyperref[tab:#1]{Tab.\,\ref*{tab:#1}}} 
\newcommand{\Thm}[1]{\hyperref[thm:#1]{Theorem\,\ref*{thm:#1}}} 
\newcommand{\Fact}[1]{\hyperref[fact:#1]{Fact\,\ref*{fact:#1}}} 
\newcommand{\Lem}[1]{\hyperref[lem:#1]{Lemma\,\ref*{lem:#1}}} 
\newcommand{\Prop}[1]{\hyperref[prop:#1]{Prop.~\ref*{prop:#1}}} 
\newcommand{\Cor}[1]{\hyperref[cor:#1]{Corollary~\ref*{cor:#1}}} 
\newcommand{\Conj}[1]{\hyperref[conj:#1]{Conjecture~\ref*{conj:#1}}} 
\newcommand{\Def}[1]{\hyperref[def:#1]{Definition~\ref*{def:#1}}} 
\newcommand{\Alg}[1]{\hyperref[alg:#1]{Alg.~\ref*{alg:#1}}} 
\newcommand{\Clm}[1]{\hyperref[clm:#1]{Claim~\ref*{clm:#1}}} 
\newcommand{\Step}[1]{\hyperref[step:#1]{Step~\ref*{step:#1}}} 
\newcommand{\Exa}[1]{\hyperref[exa:#1]{Example~\ref*{exa:#1}}} 
\newcommand{\Obs}[1]{\hyperref[obs:#1]{Obs~\ref*{obs:#1}}} 

\usepackage{enumitem}
\usepackage[nameinlink]{cleveref}
\usepackage{times}

\newcommand\pink{\ensuremath{\textrm{pink}\xspace}}

\def\supS{\textcircled{{\scriptsize S}}\xspace}
\def\supT{\textcircled{{\scriptsize T}}\xspace}
\begin{document}

\title{Directed Hypercube Routing, a Generalized Lehman-Ron Theorem, and Monotonicity Testing}
\author{Deeparnab  Chakrabarty\thanks{Supported by NSF-CAREER award CCF-2041920 and award CCF-2402571} \\
Dartmouth\\
{\tt deeparnab@dartmouth.edu}
\and
C. Seshadhri\thanks{Supported by NSF CCF-1740850, CCF-1839317, CCF-2402572, and DMS-2023495.} \\
University of California, Santa Cruz\\
{\tt sesh@ucsc.edu}
}

\date{}
\maketitle
\begin{center}
	\emph{Dedicated to Dana Ron for her 60th birthday}
\end{center}

\begin{abstract}
Motivated by applications to monotonicity testing, Lehman and Ron (JCTA, 2001) proved the existence of a collection of vertex disjoint paths between comparable sub-level sets in the directed hypercube.  The main technical contribution of this paper is a new proof method that yields a generalization to their theorem: we prove the existence of {\em two} edge-disjoint collections of vertex disjoint paths. Our main conceptual contributions are conjectures on directed hypercube flows with simultaneous vertex and edge capacities of which our generalized Lehman-Ron theorem is a special case. We show that these conjectures imply {\em directed isoperimetric theorems}, and in particular, the robust directed Talagrand inequality due to Khot, Minzer, and Safra (SIAM J. on Comp, 2018). These isoperimetric inequalities, that relate the directed surface area (of a set in the hypercube) to its  distance to monotonicity, have been crucial in obtaining the best monotonicity testers for Boolean functions.  We believe our conjectures pave the way towards combinatorial proofs of these directed isoperimetry theorems.
\end{abstract}

\section{Introduction}
We let $d \geq 2$ denote a natural number.
The directed $d$-dimensional hypercube graph $H$ 
has vertices $V(H)$ which correspond to bit-vectors $x\in \{0,1\}^d$,
and edges $E(H)$ corresponding to pairs of bit-vectors $(x,y)$ that differ in exactly one coordinate.
Edges point from lower Hamming weight vectors to larger ones.
We use $x_i$ to denote the $i$th coordinate of vertex $x$.
There is a natural partial order on the vertices/elements of the Boolean hypercube:
$x \preceq y$ iff $\forall i, x_i \leq y_i$. Note that the directed hypercube is precisely the Hasse
diagram of this partial order. Equivalently, one can consider the vertices as subsets of $[d]$,
and the partial order is given by containment.

Two subsets $S, T$ of $V(H)$ are called a \emph{matched pair}
if there exists a {\em bijection} $\phi: S\to T$ such that $s\prec \phi(s)$ for all $s\in S$; we denote a matched pair by $(S,T;\phi)$. 
An early writeup of Goldreich-Goldwasser-Ron~\cite{GGR97} posed a routing question, inspired by questions
in monotonicity testing, which was solved by Lehman and Ron~\cite{LR01}. (More discussion in \Sec{mono}.)
They were interested in the following natural question: given any matched pair $(S,T; \phi)$, can one find (edge or vertex) {\em disjoint} directed paths\footnote{
In their paper, Lehman and Ron consider these paths to be disjoint chains of subsets.} from $S$ to $T$?
Remarkably, they proved that if all points in $S$ (resp. $T$) have the same Hamming weight, then the answer is affirmative: one can find {\em vertex disjoint paths} from $S$ to $T$! 
More precisely, for an integer $0\leq i\leq d$, let $L_i$ denote the {\em $i$th layer} of $H$, that is, $L_i := \{x\in V(H)~:~ ||x||_1 = i\}$. 
We refer to this beautiful statement as the "Lehman-Ron (LR) theorem".

\begin{theorem}[Lehman-Ron Theorem~\cite{LR01}]\label{thm:lr}
	Fix any two integers $i < j$. Let $(S,T;\phi)$ be a matched pair with $S\subseteq L_i$ and $T\subseteq L_j$. 
	Then, there are $|S| = |T|$ {\em vertex disjoint paths} between $S$ and $T$.
    We refer to such a set of vertex disjoint paths as an \emph{LR solution}.
\end{theorem}

\noindent
\emph{Remark. 
The paths may not respect the bijection $\phi$. More precisely, the above doesn't prescribe vertex disjoint paths from $s$ to $\phi(s)$ for all $s\in S$.
There exist concrete counterexamples (one is given in Lehman and Ron's paper attributed to Dan Kleitman) for such paths. A follow-up work by
proves that {\em edge}-disjoint paths from $s$ to $\phi(s)$ don't exist either~\cite{BCG+10}. \medskip
}

We give an alternate proof of this theorem. But more importantly, we use our new proof technique
to strengthen the LR theorem.
If the terminals are at distance at least $2$, there exist \emph{two}
edge-disjoint LR solutions.

\begin{theorem} \label{thm:dist2}
Fix any two integers $i < j$ with $j-i \ge 2$. Let $(S,T;\phi)$ be a matched pair with $S\subseteq L_i$ and $T\subseteq L_j$. 
Then, there are $2$ collections of {\em vertex disjoint paths} between $S$ and $T$, such that their union is {\em edge disjoint}.	
\end{theorem}

\noindent
Lehman-Ron's proof of~\Cref{thm:lr} is by induction on $|S|$ and on the quantity $(j-i)$, the distance between the layers in which $S$ and $T$ lie.
The base case of $j-i=1$ is obvious as the bijection gives us the matching between $S$ and $T$.
The heart of the proof essentially shows the existence of a set of vertices $U$ either in layer $L_{j-1}$ or $L_{i+1}$ and two bijections $\phi':S\to U$ and $\phi'':U\to T$ such that
$(S,U;\phi')$ and $(U,T;\phi'')$ are matched pairs. This last part is a neat argument which uses Menger's theorem, which is a special case of the max-flow-min-cut theorem, on an auxiliary graph that they create. 
But how can one get \emph{two} edge disjoint LR solutions? The reader may notice that even the "base case'' of $j-i=2$, that is, when $S$ and $T$ are two levels apart
is itself non-trivial (indeed, we don't really know a much simpler way to solve this than the general case). And so, a new idea is needed to prove~\Cref{thm:dist2}.

Our proof of the Lehman-Ron theorem brings the {\em flow-cut duality} idea front and center. We note that~\Cref{thm:lr} is actually a statement about the structure of flows and cuts in the directed hypercube. More precisely, it states the existence of $|S|$ units of flow from vertices in $S$ to vertices in $T$ when all vertices have {\em vertex capacity} $1$ unit. 
We exploit the duality between cuts and flows, and more precisely the notion of {\em complementary slackness}, to give an alternate proof of the LR Theorem. 
In this flow-cut language, \Thm{dist2} states the existence of $2|S|$ units of flow when both edges and vertices have capacities ($1$ and $2$ units each, respectively). The existence of the two kinds
of capacities makes the argument slightly more involved, but the essence is still the same. For completeness, we show proofs of both \Cref{thm:lr} and \Cref{thm:dist2} in \Cref{sec:lr} and \Cref{sec:gen}, respectively.

We end our introduction with a natural conjecture that our techniques have been unable to solve. 
We discuss the connections between this conjecture and monotonicity testing in~\Cref{sec:mono}.
As the layers $L_i$ and $L_j$ move further apart, there should exist more collections of edge-disjoint LR solutions between $S$ and $T$.

\begin{conjecture}\label{conj:glr}
	Fix any two integers $i < j$ with $r := j-i$. Let $(S,T;\phi)$ be a matched pair with $S\subseteq L_i$ and $T\subseteq L_j$. 
Then, there are $r$ collections of {\em vertex disjoint paths} between $S$ and $T$, such that their union is {\em edge disjoint}.
\end{conjecture}

\noindent
\textbf{Other LR connections.} Recent work has generalized the LR theorem different directions in~\cite{BaGr+24}.
These results find vertex disjoint paths that "cover" any collection of points specifying certain properties.
Consider a subset $X$ of the hypercube that is partitioned into subsets of paths (or chains).
Meaning, we partition $X = \bigcup_i X_i$ such that, the vertices of $X_i$ can be ordered
according to $\prec$. (Moreover, this is the partition that minimizes the number of sets.)
In the vanilla LR setting, each $X_i$ is just $(s, \phi(s))$ for each $s \in S$. The main
theorem of~\cite{BaGr+24} shows that $X$ can be covered by a collection of vertex disjoint paths.
A nice implication of their result is that \Thm{lr} holds even if $S$ and $T$ were not contained
in levels, but were antichains.

We also note that the routing perspective in \Sec{mono} answers a question
of Sachdeva from a collection of open problems on Boolean functions~\cite{open14}
(Pg 19, "Routing on the hypercube"). We discuss more in \Sec{mono}.

\section{Alternate proof of the Lehman-Ron Theorem} \label{sec:lr} \label{sec:setup}

As a warm-up, we set up the main idea with a proof of the Lehman-Ron theorem.
We begin with an important definition.

\begin{definition} \label{def:cover} Given two sets $S$ and $T$ of the directed hypercube,
	the \emph{cover} graph $G_{S,T}$ is formed by the union of all paths from $S$ to $T$.
\end{definition}
\noindent
In other words, the cover graph is the subset of the hypercube, that contains all vertices $v$
such that $s \prec v \prec t$ (for $s \in S, t \in T$). The cover graph inherits "layers'' via intersection with the original hypercube layers. In particular, layer $L_i$ of the cover graph is only $S$ and layer $L_j$ is only $T$. \medskip

For the sake of contradiction, consider the {\em minimal counterexample} of \Cref{thm:lr} in terms of $|S|+|T|$. 
Consider the following flow network which contains $V(H)$ and also supernodes \supS and \supT.
\supS has a directed edge to every vertex in $s\in S$, and every vertex $t\in T$ has a directed edge to \supT.
We construct a flow network by setting the following {\em vertex} capacities to $G_{S,T}$: the supernodes have infinite capacity while
every vertex in $V(H)$ has capacity $1$. Since $(S,T)$ is a counterexample, by the theory of flows and cuts, the maximum \supS,\supT flow in this 
vertex-capacitated network is $< |S|$. And so, using flow-cut duality, we know that there exists a {\em cut} $C\subseteq V(H)$ such that 
(a) $|C| < |S|$, (b) {\em every} path from \supS to \supT contains a $C$-vertex. Call a path {\em cut-free} if it doesn't contain a vertex from $C$.
We can partition all vertices into three sets $\cS, C, \cT$, where $\cS$ contains all vertices that are reached by
a cut-free path from \supS, and vertices of $\cT$ can reach \supT by a cut-free path. In particular, there is no edge from a vertex in $\cS$ to a vertex in $\cT$;
all edges leaving $\cS$ enter $C$, and all edges entering $\cT$ originate from $C$.
We make a quick observation using the minimality of our counterexample.
\begin{lemma}\label{lem:minim-lr}
	$C$ is disjoint from $S\cup T$.
\end{lemma}
\begin{proof}
	If $C$ contains a vertex $S \cup T$, then one obtains a smaller counterexample. If $v\in C\cap S$, then 
	$S' := S\setminus \{v\}$, $T' := T \setminus \{\phi(v)\}$ and $\phi' := \phi_{|_{S'}}$ forms a matched pair $(S',T';\phi')$ which is also a counterexample:
	the cut $(C-\{v\})$ is a valid cut of value $|C| < |S|-1 = |S'|$. So, $C\cap S = \emptyset$. The proof of $C\cap T = \emptyset$ is analogous.
\end{proof}

\noindent
Our setup so far is a restructuring of the original Lehman-Ron proof. The following lemma is where we start to differ.
This lemma is a consequence of {\em complementary slackness} from the theory of linear optimization, and is the central tool for our new proof. 

\begin{lemma}\label{lem:complslack-lr}
	There exists a collection of {\em vertex disjoint} paths $\cP$ where every path $p\in \cP$ begins at a vertex in $S$ and ends at a vertex in $T$ and
	\begin{asparaitem}
		\item Every path $p\in \cP$ contains {\em exactly} one vertex in $C$. 
		\item Every vertex $v\in C$ is in {\em exactly} one path in $\cP$. 
	\end{asparaitem}
\end{lemma}

\noindent
Note that all these paths are in the cover graph $G_{S,T}$.
Using the above collection of paths, we make a key definition.
\begin{definition}\label{def:gateway-lr}
	A vertex $v$ is a {\em gateway} if (i) $v \in \cS$, (ii) $v$ doesn't lie on any path in $\cP$,
	and (iii) there is at least one edge $(v,w)$ in the cover-graph.
\end{definition}
\noindent
The Lehman-Ron theorem, \Cref{thm:lr}, follows directly from the following lemma.

\begin{lemma}\label{lem:main-lr}
	For all $i\leq k\leq j-1$, the $k$th layer $L_k$ of the cover-graph contains a gateway vertex.
\end{lemma}

\begin{proof}[\bf Proof of~\Cref{thm:lr}]
	Consider the gateway vertex $v \in L_{j-1} \cap \cS$ with edge $(v,w)$ in the cover-graph. Note $w\in T$ and therefore
    in $\cT$. There is an edge from $\cS$ to $\cT$. Contradiction. Hence, there is no (minimal) counterexample
    to \Thm{lr}.
\end{proof}
\noindent
Before giving the formal proof of~\Cref{lem:main-lr} directly, let us describe the main idea which uses the symmetry of the hypercube.
First, let us observe the layer $L_i$, that is $S$, contains a gateway vertex $s$. Indeed, there are $\leq |S|-1$ paths in $\cP$
and so there is some $s\in S$ not in any of these paths. Furthermore, all edges that lead $s$ to $\phi(s)$ lie in the cover-graph.

Now, let's see
how to get a gateway in the next layer $L_{i+1}$. Consider any edge
$(s,x^{{(1)}})$ in $G_{S,T}$, and suppose this edge corresponds to projecting according to some dimension $r$. That is $s_r = 0$ and $x^{(1)}_r = 1$.
If $x^{(1)}$ is not in any path in $\cP$, we have discovered the desired gateway in $L_{i+1}$. 
Otherwise, $x^{(1)}$ lies on some path, say, $P\in \cP$. Follow $P$ {\em forwards} for a single edge from $x^{(1)}$, to get to $x^{(2)}$.
Note that $x^{(2)} \in L_{i+2}$ and has $r$th-coordinate $1$. Now, one can project "down" on the $r$-coordinate
to get $x^{(3)}\in L_{i+1}$. Observe that $(s, x^{(3)})$ is a {\em projection} of the edge $(x^{(1)},x^{(2)})$ along the $r$th-dimension
and hence is an edge of the cover-graph. Next, observe that $x^{(3)}$ cannot be in $\cT$, so either $x^{(3)} \in \cS$ or $x^{(3)} \in C$. If $x_3 \in \cS$
and not on any path in $\cP$, we are done. Otherwise $x^{(3)}$ lies on some other path $Q \in \cP$. 
We now walk {\em backwards} along $Q$, to get $x^{(4)} \in S$. Noting that $x^{(4)}$ has $r$-coordinate $0$,
we can redo the entire process above. Observe that each "step" proceeds along a \emph{matching}. Either we project, walk from $L_{i+1}$ to $L_{i+2}$
using a path in $\cP$, or walk backward from $L_{i+1}$ to $L_{i}$ using a path in  $\cP$. Each of these is using a matching edge,
and no vertex is ever visited twice in the entire process. Hence, this process must terminate, at which point a gateway is discovered. 
And then one uses the same idea to obtain a gateway vertex in $L_{i+2}$, and so on.
One can convert this idea into a formal proof, but it becomes notationally cumbersome. 
A cleaner proof method is to consider a potential "fixed point" of this process, and prove a contradiction. 

\begin{proof}[\bf Proof of~\Cref{lem:main-lr}]
    Fix a collection of paths $\cP$ as given by \Cref{lem:complslack-lr}.
	As argued above, $L_i$ has a gateway vertex.
	Let $k \in [i,j-1]$ be the largest value such that $L_k$ contains a gateway vertex. If $k = j-1$, we are done. So suppose
	that $k < j-1$. We now engineer a contradiction.
	Let $v^\star$ be a gateway in $L_k$. So $v^\star \in \cS$ and has an edge $(v^\star,w)$ leaving it. Let $r$ be the dimension of this edge
	implying $v^\star_r = 0$ and $w_r = 1$. Let $\Pi_r$ be the projection operator which flips the $r$th coordinate; so $\Pi_r(v^\star) = w$ and
	vice versa. Define the following sets; we give a illustration for convenience where the pink highlighted edges participate in paths of $\cP$.
	
	\begin{minipage}{0.5\textwidth}
		\begin{asparaitem}
			\item $A = \{a ~:~ a \in \cS \cap L_k, a_r = 0 ~\text{and}~ \Pi_r(a)\in G_{\cS,\cT}\}$. 
			\item $X = \Pi_r(A) = \{\Pi_r(a)~:~a\in A\}$. 
			\item $B = \{b ~:~ (x,b) \in P,~\text{for some} ~x \in X, P\in \cP\}$. 
			\item $Y = \Pi_r(B) = \{\Pi_r(b)~:~b\in B\}$. 
		\end{asparaitem}
	\end{minipage}
	\begin{minipage}{0.45\textwidth}
		\includegraphics[trim = 150 0 0 0, clip, scale=0.35]{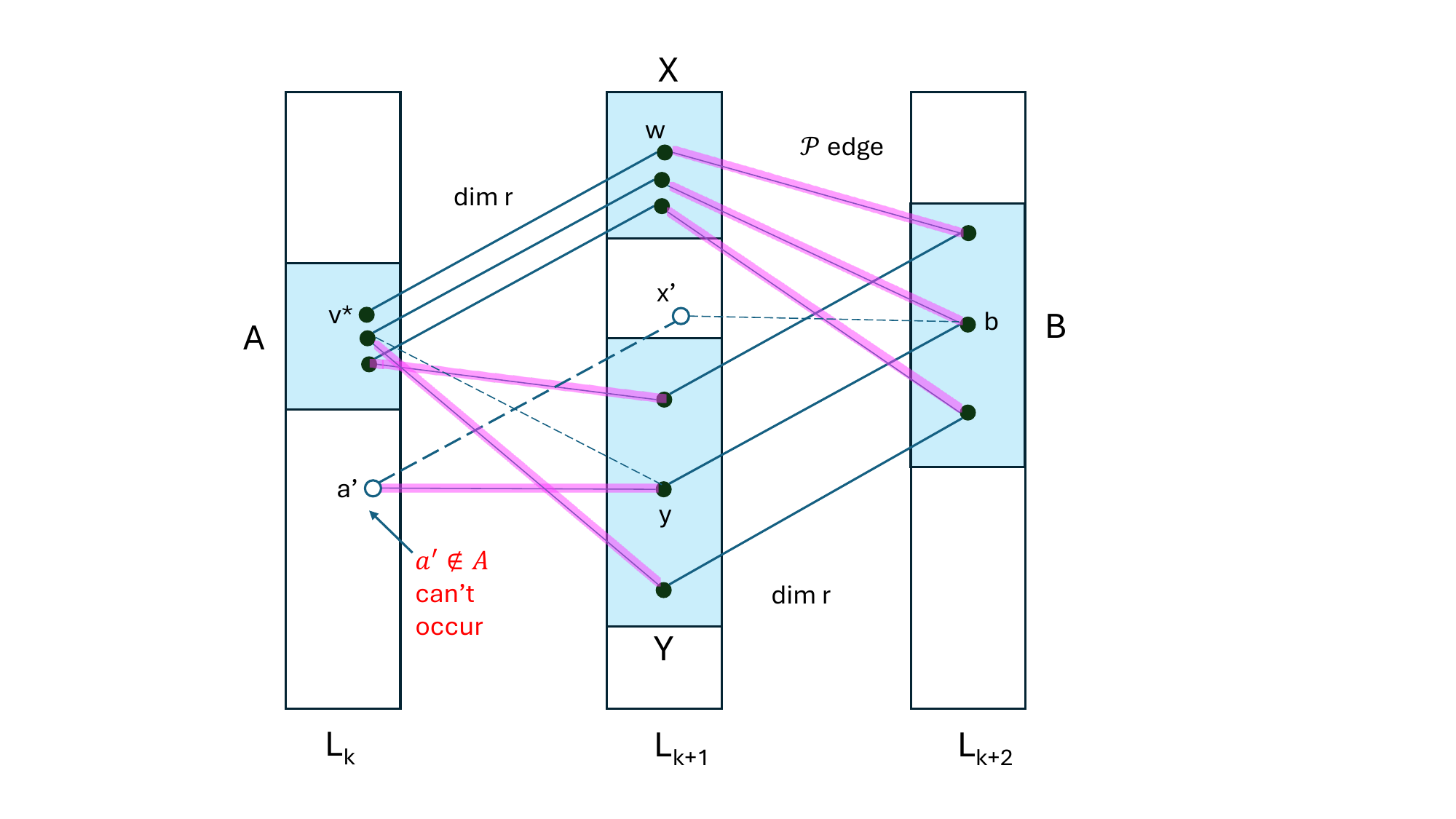}
	\end{minipage}		
	
	\noindent
	Observe that $X \cup Y \subseteq L_{k+1}$, and $B \subseteq L_{k+2}$. We next argue these subsets lie in the cover-graph.
	By definition of $A$, $X\subseteq G_{\cS,\cT}$, and since $B$ is obtained via paths from $X$, we get $B\subseteq G_{\cS,\cT}$.
	Finally, for any $y\in Y$, note that $b = \Pi_r(y)$ lies in $B$, and there's a path $(a,x, b)$ for some $a\in A$ and $x = \Pi_r(a)$.
	But note that $(a,y)$ edge is the $r$-projection of $(x,b)$, and so $(a,y,b)$ is a path implying $y\in G_{\cS,\cT}$. Thus, $Y$ also lies in the cover-graph.
	
	Since $X$ is a projection of $A$, we get $|X| = |A|$.
	Consider any $x\in X$. Since $x\in L_{k+1}$, by our choice of $k$, $x$ isn't a gateway vertex. Since $(a,x)$ is an edge for $a = \Pi_r(x)$, 
	we have $x\in \cS$ or $x\in C$. In the former case since $x$ isn't a gateway vertex, and in the latter case by~\Cref{lem:complslack-lr}, there
	exists $P\in \cP$ with $x\in P$. And so,  $|B| = |X|$. By projection property, $|Y| = |B|$, and so following the above chain of equalities, we get $|Y| = |A|$.
	
	Next consider any $y\in Y$. As argued above, there is a vertex $a\in A$ such that $(a,y)$ is an edge in the cover-graph (the projection of $(x,b) \in P$ where $b=\Pi_r(y)$).
	Since $y\in L_{k+1}$, by our choice of $k$, $y$ isn't a gateway vertex. Since $(a,y)$ is an edge, we must have $y\in \cS$ or $y\in C$. 
	In the former case since $y$ isn't a gateway vertex, and in the latter case by~\Cref{lem:complslack-lr}, there
	exists $Q\in \cP$ with $y\in Q$. Let $(a',y)$ be the edge in $Q$ taking $y$ "backward'' along $Q$. We claim that $a' \in A$. 
	If $y\in \cS$, then $a'\in \cS$ since $(a',y)$ is an edge; if $y\in C$
	then by~\Cref{lem:complslack-lr}, $a'\in \cS$ (the path $Q$ doesn't contain two cut-vertices). Since $(a',x' = \Pi_r(a'))$ lies in the cover graph since the $r$-projection of the $(a',y)$ lies there, 
	we get $x'$ lies in the cover-graph, implying $a'$ must lie in $A$. The above figure illustrates this.
	In short, there are $|Y|$ paths of $\cP$ that contain vertices of $A$. Since $|Y| = |A|$ and since these paths are all vertex-disjoint, we conclude {\em all} vertices in $A$ are present in some 
	path of $\cP$. However, the gateway vertex $v^\star \in A$ and doesn't lie on any path of $\cP$. Contradiction.
\end{proof}

\section{Proof of the generalization~\Cref{thm:dist2}}\label{sec:gen}
We now prove the generalization of the Lehman-Ron theorem using the proof strategy above.  As before, we start with a minimal counter-example and use it to construct a flow network.
We have the same directed graph as in the previous section with supernodes \supS and \supT, with every edge incident to supernodes having infinite capacity.
The crucial difference is that we have both vertex and edge capacities. Each edge has capacity one, and each vertex has capacity two.
(Alternately, it costs "one unit" to cut an edge, but "two units" to cut a vertex.)

Since we have a counter-example, again using the theory of flows, the maximum flow in this network is $< 2|S|$. 
And thus, by duality, we posit that there exists a {\em pair} $(C,F)$ with $C\subseteq V(H)$ and $F\subseteq E(H)$, such that (a) $2|C| + |F| < 2|S|$, (b) every path from $s$ to $t$ either contains a $C$-vertex or an $F$-edge or both. A path from a vertex $u$ to $v$ is now called {\em cut-free} if it contains neither a vertex from $C$ nor an edge from $E$.
As before, we can partition all vertices into three sets $\cS, C, \cT$, where $\cS$ contains all vertices that are reached by
a cut-free path from $S$, and vertices of $\cT$ can reach $T$ by a cut-free path. Note that the edges of $F$ are from vertices in $\cS$ to vertices in $\cT$.
And, as before, by minimality of the counter-example, the following simple observation holds.

\begin{lemma}\label{lem:minim} (i) The cut set $C$ is disjoint from $S \cup T$. (ii) There exists a mincut $(C,F)$ such that no vertex participates in more than one edge of $F$.
\end{lemma}
\begin{proof}
	Proof of (i) is exactly as in~\Cref{lem:minim-lr}.
	Suppose $v$ participates in at least two edges of $F$. Observe that any $S$-$T$ path through any of these edges must go via $v$.
	Hence, we can remove these edges from $F$, add $v$ to $C$, and preserve the fact that $C \cup F$ is an $S$-$T$ cut. Moreover,
	the cut value does not increase.
\end{proof}
\noindent
We can now give the analog of~\Cref{lem:complslack-lr}.
\begin{lemma}\label{lem:complslack}
	There exists a collection of paths $\cP$ with the following properties.
	\begin{asparaitem}
		\item Every path $p\in \cP$ begins at a vertex in $S$ and ends at a vertex in $T$.
		\item The paths are pairwise edge-disjoint and any vertex is in at most two paths.
		\item Every path $p\in \cP$ either contains {\em exactly} one vertex in $C$ or {\em exactly} one edge in $F$, but not both.
		\item Every vertex $v\in C$ is in {\em exactly} two paths in $\cP$ and every edge $e\in F$ is in {\em exactly} one path of $\cP$.
	\end{asparaitem}
\end{lemma}
\begin{proof} 
	By complementary slackness (Theeorem A.7 of~\cite{Bills-book}) , every maximum flow must saturate the min cut. Together with integrality of flow, 
	this implies the existence of an integral flow saturating $C \cup F$.
	We give a (simple) formal explanation. By integrality of flow, there is a maximum flow that can be decomposed into paths. Let $\cP$ be those paths.
	Since these paths form a feasible flow, they satisfy the first two bullet points of the lemma. By duality, $|\cP| = 2|C| + |F|$.
	For each path $p \in P$, let $c_p$ be the number of cut elements in $C \cup F$ that the path contains. For each cut element
	$e \in C \cup F$, let $k_e$ be the number of paths that $e$ participates in. So $\sum_{p \in \cP} c_p = \sum_{e \in C \cup F} k_e$.
	Note that $\forall p, c_p \geq 1$, since $C \cup F$ is a valid cut. Thus, $\sum_{p \in \cP} c_p \geq |\cP|$.
	Now, observe that $\forall e \in C, k_e \leq 2$ and $\forall e \in F$, $k_e \leq 1$, since $C \cup F$ must satisfy the flow
	constraints. Hence, $\sum_{e \in C \cup F} k_e \leq 2|C| + F$.
	We get 
	$ |\cP| \leq \sum_{p \in \cP} c_p = \sum_{e \in C \cup F} k_e = 2|C| + |F| = |\cP|$.
	Thus, the inequalities above are all equalities. So $\forall p, c_p = 1$ (third bullet)
	and $\forall e \in C, k_e = 2$ and $\forall e \in F$, $k_e = 1$ (fourth bullet).
\end{proof}
\noindent
Next we provide the relevant generalization of gateway vertices earlier defined in~\Cref{def:gateway-lr}
\begin{definition}\label{def:gateway}
	A vertex $v$ is a {\em gateway} if (i) $v \in \cS$, (ii) $v$ lies on at most one path in $\cP$,
	and (iii) there is at least one edge $(v,w) \notin F$ leaving $v$ in $G_{\cS, \cT}$. 
\end{definition}
\noindent
As before, the proof of~\Cref{thm:dist2} follows from the following lemma, and the remainder of this section will prove it.
\begin{lemma}\label{lem:main}
	For all $i\leq k\leq j-1$, $L_k$ contains a gateway vertex.
\end{lemma}

\noindent
As in the proof of~\Cref{lem:main-lr}, we proceed via minimal counterexamples. 
First, we establish that there is a vertex in $L_i$ that is a gateway vertex. There are $< 2|S|$ paths in $\cP$.
Since any vertex participates in at most $2$ paths, some vertex in $s$ participates in at most $1$ path.
Since $\phi(s)$ is at least distance $2$ away from $s$, $s$ has degree at least $2$ in $G_{S,T}$. 
(This is where the distance between $S$ and $T$ is used.) At most one of those edges is in $F$ (\Lem{minim}),
so there is some edge leaving $s$ that is not in $F$. Thus, $s$ is a gateway vertex.  \medskip

Now, let $k \in [i,j-1]$ be the largest
value such that $L_k$ contains a gateway vertex. If $k = j-1$, we are done. So suppose that $k < j-1$, and  we will engineer a contradiction.
Let $v^\star$ be a gateway vertex in $L_k$. So $v^\star \in \cS$, participates in at most one path of $\cP$,
and there is some edge leaving $v^\star$ that is not in $F$. First, let us choose the projection dimension $r$ as follows.
If $v^\star$  lies on exactly one path  $P \in \cP$, then let $(v^\star, w)$ be the edge of $P$ incident to $v^\star$
and let $r$ be the coordinate that this edge flips. If $v^\star$ lies on no path $P\in \cP$, then
since $k < j-1$, there are at least two\footnote{there is some $s\in S$ such that $s \prec v \prec \phi(s)$ and there are at least $2$ edge disjoint paths from $v$ to $\phi(s)$} edges incident to $v^\star$ that are in the cover-graph; let $(v^\star, w)$ be any such edge and let $r$ be the coordinate of this edge.
As in the proof of~\Cref{lem:main-lr}, we let $\Pi_r$ denote the projection operator, and we define the sets exactly as in the previous proof.

\begin{asparaitem}
	\item $A = \{a ~:~ a \in \cS \cap L_k, a_r = 0 ~\text{and}~ \Pi_r(a)\in G_{\cS,\cT}\}$. Note that $v^\star \in A$.
	\item $X = \Pi_r(A) = \{\Pi_r(a)~:~a\in A\}$. 
	\item $B = \{b ~:~ (x,b) \in P,~\text{for some} ~x \in X, P\in \cP\}$. 
	\item $Y = \Pi_r(B) = \{\Pi_r(b)~:~b\in B\}$. 
\end{asparaitem}

\begin{figure}\label{fig:covgraph}
	\centering
	\includegraphics[trim = 0 60 0 30, clip, scale=0.3]{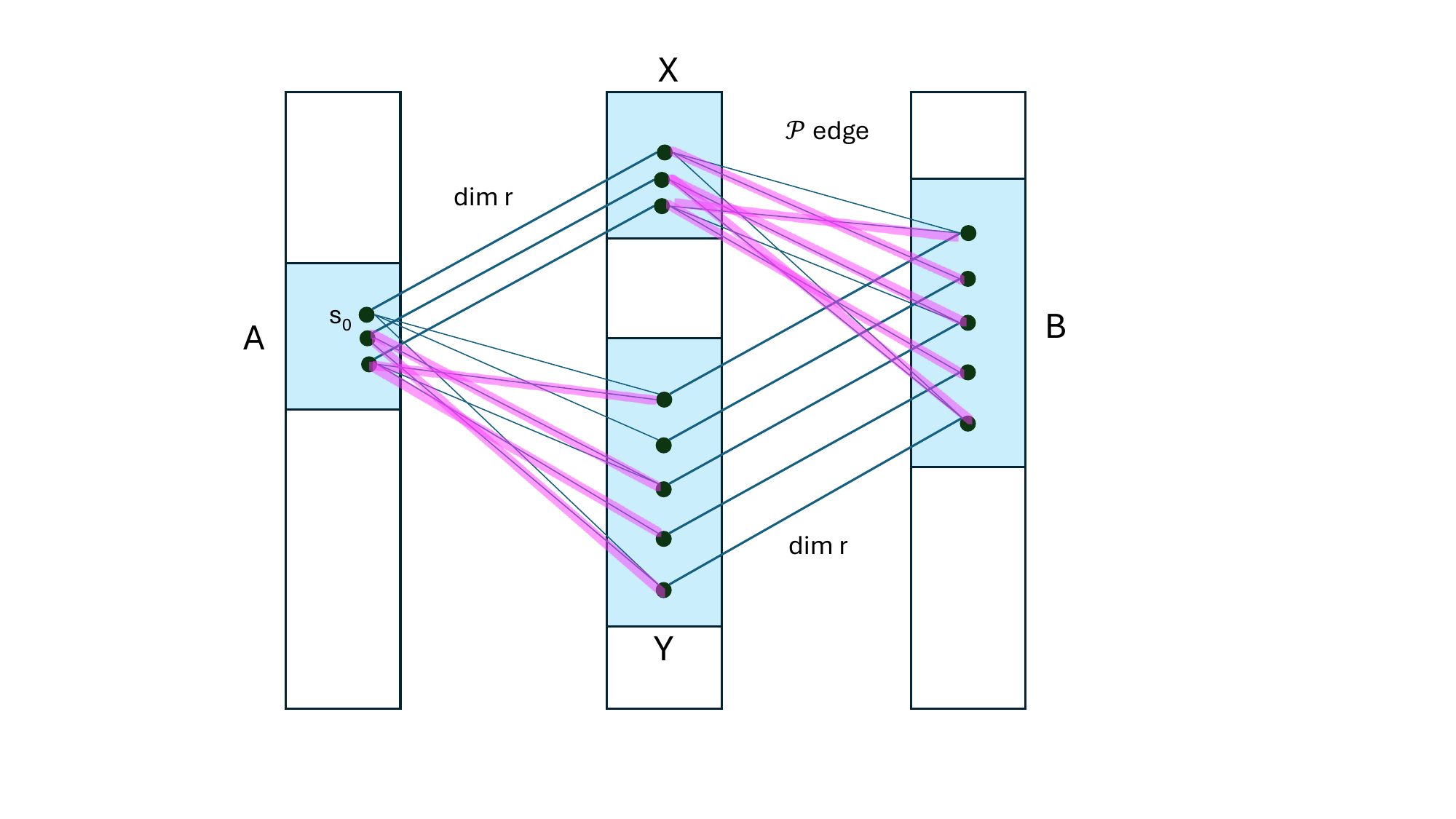}
	\caption{The setup of the proof of \Lem{main}. The set $X$ is a projection of $A$
		along dimension $r$, and similarly $Y$ is a (downward) projection of $B$.
		The edges between $A$ and $Y$ are projections of edges between $X$ and $B$.}
\end{figure}

\noindent
As in the proof of~\Cref{lem:main-lr}, we have $X \cup Y \subseteq L_{k+1}$, and $B \subseteq L_{k+2}$ and that $A, X, B, Y$ 
all lie in the cover graph $G_{\cS, \cT}$. As there, we crucially use the property that the projection of every $(A,Y)$ edge lies is an $(X,B)$ edge.
Note that, we have $|A| = |X|$ and $|B| = |Y|$. However, unlike in~\Cref{lem:main-lr}, there can be {\em two} paths of $\cP$ that may contain $x\in X$, and thus, $|B|$ may be larger in size than $|X|$.
And this makes the proof slightly  more complicated. 

We now restrict our attention to these four sets. As a result, we only consider the 
cover graph $G_{A,B}$, which, recall, is the union of all paths from vertices in $A$ to vertices in $B$ (irrespective of whether they contain cut-elements).
Both $A$ and $B$ lie in this graph by definition. We assert $X\subseteq G_{A,B}$. Take a vertex $x\in X$. If $x \in P$ for some path $P\in \cP$ then it lies in $G_{A,B}$ by construction.
We now show if $x$ is not any path in $\cP$, we already have a contradiction. Indeed, the edge $(a,x)$ where $a = \Pi_r(x)$ isn't in $F$ for otherwise, by~\Cref{lem:complslack} there would be a path containing $(a,x)$. Similarly $x\notin C$. So, $x\in \mathcal{S}$. Since $x\in L_{k+1}$ and $k < j-1$, there is at least one edge leaving it, and as before, this edge isn't in $F$. So, $x$ is a gateway vertex which contradicts the choice of $k$. Thus, the subset $X$ lies in $G_{A,B}$. And since every $(A,Y)$ edge is an  $r$-projection of an $(X,B)$ edge, we argue that $Y$ also lies in $G_{A,B}$.

\noindent
We now make another definition which will lead to our contradiction.

\begin{definition} \label{def:pink} An edge in $G_{A,B}$ is called \emph{pink}
	if it lies in $\cP$ \emph{and} does not change the $r$ coordinate. 
	For any subset of vertices $W$ in $G_{A,B}$, $\pink(W)$ is the number
	of pink edges $W$ is incident to. For singleton subsets $\{v\}$ we abuse notation
	and call $\pink(\{v\})$ simply $\pink(v)$.	
\end{definition}
\noindent
Thus, every pink edge incident to $A$ is incident to $Y$,
and every pink edge incident to $X$ is incident to $B$. 
The following claim is a consequence.

\begin{claim} \label{clm:equal} $\pink(A) = \pink(Y)$ and $\pink(X) = \pink(B)$.
\end{claim}

Next, observe that the gateway $v^\star$ participates in at most one path of $\cP$.
If such a path existed, we chose the dimension $r$ according to the incident
edge of this path. Hence, $\pink(v^\star) = 0$.
This is central to proving the final contradiction.

\begin{claim} \label{clm:ax} $\pink(X) > \pink(A)$.
\end{claim}

\begin{proof} There is a perfect matching between $A$ and $X$. 
	We prove that $\forall a \in A$, $\pink(\Pi_r(a)) \geq \pink(a)$.
	Furthermore, for the gateway $v^\star$, we get a strict inequality $\pink(\Pi_r(v^\star)) > \pink(v^\star)$.
	It is convenient to do a case analysis based on whether $x = \Pi_r(a)$ is in $\cS, C$, or $\cT$.
	Note that $\pink(a) \leq 2$ by \Lem{complslack} since any vertex participates in at most $2$ paths of $\cP$. Note that all edges of $\cP$ that leave 
	$X$ lead to $B$ (by construction), so all these edges are in $G_{A,B}$.
	
	{\em $x \in C$:} This is the easiest. By \Lem{complslack}, there are two paths through $x$. So $\pink(x) = 2 \geq \pink(a)$.
	
	{\em $x \in \cT$:} Note that $a \in \cS$, and so in this case, the $r$-projection edge $(a,x)$ must be $\in F$. By \Lem{complslack}, there is 
	a path of $\cP$ through $(a,x)$, and the next edge leaving $x$ goes into $B$. Since the $r$-coordinate doesn't change, this edge is pink, and so $\pink(x) \geq 1$.
	Since $a$ participates in the path through projection edge $(a,x)$, which is {\em not} pink, it can participate in at most one other path of $\cP$
	(by \Lem{complslack}, no vertex is in more than $2$ paths). Thus, in this case $\pink(a) \leq 1$ implying $\pink(a) \leq \pink(x)$.
	
	{\em $x \in \cS$:} By our assumption, $x$ cannot be a gateway vertex.
	So either all edges leaving $x$ are in $F$ or $x$ is in $2$ paths of $\cP$. If the latter happens,
	then $\pink(x) = 2$, completing this case. 	
	So let us assume the former. There must be some
	edge leaving $x$ in $G_{A,B}$ because $x$ is in the cover graph $G_{A,B}$; since $x$ is not a gateway vertex, this edge which must be in $F$.
	By~\Cref{lem:complslack}, there is a path $P\in \cP$ containing this and the edge incident to $x$ doesn't change the $r$-coordinate, and thus is pink.
	So, $\pink(x) \geq 1$.  So, if $\pink(a) \leq 1$ we are done. So, suppose 
	$\pink(a) = 2$. That is, there are two distinct edges $(a,y)$ and $(a,y')$ which are pink. 
	Consider the $r$-projection of these edges, $(x, \Pi_r(y))$ and 
	$(x, \Pi_r(y'))$; these are present in $G_{A,B}$ and so by our assumption above, these two must be in $F$.
	But again by~\Cref{lem:complslack}, there are two paths in $\cP$ containing these, and so these edges are pink, implying $\pink(x) = 2$ as well.
	This settles this case.

	All in all, we have proven that $\pink(x) \geq \pink(a)$, and note that in all cases, $\pink(x) \geq 1$.
	Since $\pink(v^\star) = 0$, we get the strict inequality $\pink(\Pi_r(v^\star)) > \pink(v^\star)$, completing the proof of the claim.
\end{proof}
\noindent
The next claim which, along with \Clm{equal}, contradicts \Clm{ax}, completing our proof of~\Cref{lem:main}.

\begin{claim} \label{clm:by} $\pink(B) \leq \pink(Y)$.
\end{claim}

\begin{proof}
	There is a perfect matching between $B$ and $Y$ by the $\Pi_r$ projection. We will show that
	for every $b \in B$, $\pink(b) \leq \pink(y)$, where $y = \Pi_r(b)$. 
	The proof is analogous to that of \Clm{ax}, with a subtle difference. All edges of $\cP$
	leaving $X$ are in $G_{A,B}$ by construction. But all edges of $\cP$ entering $Y$ might 
	not be in $G_{A,B}$. In particular, there could be a path $Q\in \cP$ which contains $y$
	and the predecessor, call it $z$, of $y$ on this path may not be in $A$. For instance, this could occur 
	if $y \in \cT$ and $z \in C$ (and thus not in $A$). 
	
	With hindsight, we assert that 
	if $y \in \cS \cup C$, then the vertex $z$ indeed lies in $A$. To see this note that $z\in \cS$;
	this follows from \Lem{complslack} since the path $Q$ contains exactly one cut-vertex or cut-edge.
	Furthermore, $\Pi_r(z)$ must lie in $G_{\cS,\cT}$ since $(z,\Pi_r(z), b)$ is present in the hypercube.
	In all, $z\in A$.
	The reader may wonder why this subtlety didn't arise in the original Lehman-Ron proof that we showed in the previous section.
	Well, there is a vertex $a\in A$ such that $(a,y)$ is an edge, and since in the previous section we only had vertex cuts, we could (and did) assert $y\in S\cup C$.
	But now the edge $(a,y)$ could be in $F$. As we will see, the case when $y\in T$ is actually easy to take care of.
	
	Now for the case analysis. Recall that $b \in B$ and $y = \Pi_r(w) \in Y$.
	
	{\em $y \in C$:} By~\Cref{lem:complslack} there are two paths entering $y$, and since $y\in C$ due to the discussion about the subtlety above, $\pink(y) = 2$.
	Note that $\pink(b) \leq 2$ since there can be at most two paths of $\cP$ incident on $b$.
	
	{\em $y \in \cS$:} Since $y$ is not a gateway vertex (overall assumption), either $y$
	participates in two paths of $\cP$ or all edges leaving $y$ are in $F$. 
	In the former case, since $y\in S$ due to the discussion about the subtlety above, $\pink(y) = 2$. In the latter case, the edge $(y,b)$ must be in $F$.
	So there is at least one path through $y$ and, again since $y\in S$, $\pink(y) \geq 1$. Now note that the $(y,b)$ edge is {\em not} pink although it is on a path in $\cP$ 
	because the $r$-coordinate changes. So, $\pink(b) \leq 1$.
	
	{\em $y \in \cT$:} Observe that all edges $(z,y)$ with $z \in A \subseteq \cS$ must be cut, that is $(z,y) \in F$.
	By \Lem{minim}, all such edges are on paths in $\cP$ and are thus pink (they are clearly in $G_{A,B}$ and don't change $r$-coordinate).
	Suppose $\pink(b) = t$ where $t\in \{1,2\}$. Then there are $t$ distinct pink edges of the form $(x,b)$ with these $x$'s in $X$.
	Consider their $r$-projections, that is, the $t$ edges of the form $(\Pi_r(x), y)$. Since $\Pi_r(x)\in A$, all these edges must be pink.
	This shows that if $y\in \cT$, $\pink(y) \geq \pink(b)$.
\end{proof}

\section{Connections to monotonicity testing} \label{sec:mono}

The motivation for \Conj{glr} (and indeed, \Thm{lr}) is a deeper understanding of the problem of monotonicity testing of functions, a problem which,
especially over the hypercube and hypergrid domains, has had a rich history of more than 25 years~~\cite{Ras99,EKK+00,GGLRS00,DGLRRS99,LR01,FLNRRS02,HK03,AC04,HK04,ACCL04,E04,SS08,Bha08,BCG+10,FR,BBM11,RRSW11,BGJ+12,ChSe13,ChSe13-j,ChenST14,BeRaYa14,BlRY14,ChenDST15,ChDi+15,KMS15,BeBl16,Chen17,BlackCS18,BlackCS20,BKR20,HY22,BrKh+23,BlChSe23, BlChSe23-2}. A function $f:\{0,1\}^d \to \{0,1\}$ is monotone if $\forall x,y \in \{0,1\}^d$ where $x \prec y$, $f(x) \leq f(y)$.
The distance between two functions $f,g$ is $|\{x: f(x) \neq g(x)\}|/2^d$, and the distance of $f$ to monotonicity, denote $\eps_f$, is the minimum distance of $f$ to a monotone $g$.
A function $f$ is said to be $\eps$-far from monotone if $\eps_f \geq \eps$.
The aim of a tester is to distinguish a monotone function from one that is "far" from monotone. 
There is a special focus on non-adaptive monotonicity testers with one-sided error. These are testers that (i) always accept monotone functions, and (ii) make all their queries in advance.
After a long line of work, this has been resolved (up to $\log d, \poly(\eps^{-1})$ factors) to
be $\Theta(\sqrt{d})$~\cite{GGLRS00,FLNRRS02,ChSe13-j,ChenST14,KMS15,ChenDST15,Chen17}.

The study of these monotonicity testers led to the discovery of \emph{directed isoperimetric inequalities}.
Much of the study in this paper came from attempts at an alternate, more combinatorial proofs of a 
central isoperimetric inequality, the so-called {\em robust directed Talagrand theorem} due to Khot, Minzer, and Safra~\cite{KMS15}.
In this section, we give connections between monotonicity testing, directed isoperimetric inequalities
(like the KMS theorem), and routing on the directed hypercube. Most importantly, we describe another 
routing statement, \Conj{rout}, which implies the KMS theorem. We believe that \Conj{rout} and \Conj{glr}
are closely related, as explained in \Sec{conj}.

\begin{definition} \label{def:viol} A pair of vertices of $\{0,1\}^d$ $x \prec y$
is called a \emph{violation} of $f(x) > f(y)$. This pair is called a \emph{violated edge}
if additionally, $(x,y)$ is an edge of the hypercube.
For any $x$, the directed influence $\dirinf{f}{x}$ is the number of violated edges incident to $x$.
The directed influence of $f$, denoted $I^+_f$, is $2^{-d} \sum_x \dirinf{f}{x}$.
\end{definition}
\noindent
The most basic inequality is the directed Poincare inequality, which directly
leads to $O(d)$ query monotonicity testers (for constant $\eps$).

\begin{theorem} \label{thm:dirpoin} \cite{GGLRS00} For any $f:\{0,1\}^d\to \{0,1\}$,  $I^+_f \geq \eps_f$.
\end{theorem}

The main step towards $o(d)$ query testers (for constant $\eps$) is a stronger isoperimetric inequality,
the directed Margulis bound. Let $\Gamma^+_f$ denote the size of the largest matching
of violated edges, which is a measure of the vertex boundary.

\begin{theorem} \label{thm:dirmarg} \cite{ChSe13-j} For any $f:\{0,1\}^d\to \{0,1\}$, $I^+_f \cdot \Gamma^+_f = \Omega(\eps^2_f)$.
\end{theorem}

The culmination of this line of work lead to the \emph{robust, directed Talagrand inequality} for KMS,
which yielded the (near) optimal $\widetilde{O}(\sqrt{d})$-query non-adaptive monotonicity tester. (The original KMS result
lost a log factor, which was removed by Pallavoor-Raskhodnikova-Waingarten~\cite{PRW22}.)

\begin{theorem} \label{thm:dirtal} \cite{KMS15,PRW22} Let $\chi$ be any bicoloring of the directed
hypercube edges, with two colors $0$ and $1$. For any $f:\{0,1\}^d\to \{0,1\}$ and for any $x\in \{0,1\}^d$, let $\dirinfc{f}{\chi}{x}$
be the number of violated edges incident to $x$ whose color is $f(x)$. Then,

\[
2^{-d} \sum_x \sqrt{\dirinfc{f}{\chi}{x}} = \Omega(\eps_f)
\]
\end{theorem}

\noindent
In the next subsection, we give combinatorial  interpretations to each of these statements.
The reason for \Conj{glr} and a deeper study of hypercube routing was to get alternate proofs of \Thm{dirtal}.
A big mystery of all these directed isoperimetric inequalities is the appearance of $\eps_f$,
the distance to monotonicity, as a "directed version" of the variance of $f$. It appears
as if $\eps_f$ is the "right measure" of directed volume. We hope that alternate proofs of \Thm{dirtal}
may shed some light on this mystery.

\subsection{From flows to directed isoperimetry} \label{sec:flow}

In what follows, all flow networks are over the directed hypercube. 
There is a source set $S$,
and the aim is to route flow to the complement $\overline{S}$\footnote{Formally, one creates
a supernode $\supS$ that connects to $S$, and a supernode $\supT$ with connections from $T$.
All these connections have infinite capacity.}.
In the various routing theorems, we set different edge/vertex capacities and try to lower bound
the maximum flow from $S$ to $\overline{S}$. In all the flow settings, we have unit edge
capacities. 

\noindent
We start with a notion of the "directed volume" of a set.

\begin{definition} \label{def:dirvol} For $S \subseteq \{0,1\}^d$, the
directed volume of $S$, denoted $\dirvol{S}$ is 
$$ \max_{\substack{S' \subseteq S\\ T' \subseteq \overline{S}}}\Big\{ |S'| \ \Big| \ \exists \phi, (S', T'; \phi) \ \textrm{is a matched pair}\Big\} $$
Any matched pair $(S',T';\phi)$ that attains the maximum is called a \emph{directed volume certificate}.
\end{definition}
\noindent
We now explain why the directed Poincare inequality of \Thm{dirpoin} essentially shows that one can send $\dirvol{S}$ units of flow
from $S$ to $\overline{S}$ with unit edge capacities. This is a simple application of the max-flow-min-cut theorem, and we 
provide the proof for completeness.

\begin{theorem} \label{thm:flowpoin} Consider the directed hypercube flow network with unit edge capacities,
and source set $S$. The maxflow is at least $\dirvol{S}$.
\end{theorem}

\begin{proof} Consider the indicator Boolean function $f:\{0,1\}^d \to \{0,1\}$ where $f(x) = 1$ iff $x \in S$.
Using a standard connection between distance to monotonicity (Corollary 2 of~\cite{FLNRRS02}) one can 
argue that $\eps_f = \dirvol{S}/2^d$. 
Any $\supS$-$\supT$ cut must remove all edges from $S$ to $\overline{S}$. 
These are precisely the violated edges of $f$, which are at least $\eps_f 2^d = \dirvol{S}$ many (\Thm{dirpoin}).
The theorem follows from the duality between max-flow and min-cut. 
\end{proof}
\noindent
Thus, the basic directed Poincare inequality basically gives a flow bound for the directed hypercube
flow network. We will now interpret the more sophisticated isoperimetric theorems as more general
flow statements.
A crucial notion is the separation distance of a set.

\begin{definition} \label{def:sep} For any matched pair $(S,T;\phi)$, the separation distance
is $|S|^{-1} \sum_{s \in S} (|\phi(s)| - |s|) = |S|^{-1} \Big(\sum_{t \in T} |t| - \sum_{s \in S} |s|\Big)$.
Here $|x|$ denotes the number of $1$s in $x\in \{0,1\}^d$.
The {\bf separation distance} of $S$ is the smallest separation distance over directed volume certificates of $S$.
\end{definition}

A key theorem of \cite{ChSe13-j} shows that larger separation distance implies more (edge disjoint) flow.
This theorem is a strengthening of the directed Poincare inequality of \Thm{dirpoin},
and essentially a flow rewording of Lemma 2.6 of \cite{ChSe13-j}. The proof is entirely analogous
to that of \Thm{flowpoin} and is omitted.

\begin{theorem} \label{thm:cspoin} [Lemma 2.6,~\cite{ChSe13-j}] Consider the directed hypercube flow network with unit edge capacities, with source
set $S$ having separation distance $r$. The maxflow is at least $r \dirvol{S}$.
\end{theorem}

What if we desire \emph{vertex} disjoint paths? Lemma 2.5 of~\cite{ChSe13-j}
answers this question, and the central tool is the Lehman-Ron theorem.
Together, the two theorems above directly imply the directed Margulis statement of \Thm{dirmarg}.
\begin{theorem} \label{thm:cslr} [Lemma 2.5,~\cite{ChSe13-j}] Consider a flow network with unit vertex capacities,
with source set $S$ having separation distance $r$. The maxflow is at least $\dirvol{S}/32r$.
\end{theorem}

\noindent
This brings us to a sort of "intellectual starting point" for this paper. The theorems
above clearly show how directed isoperimetry and flows are intimately connected. Moreover,
statements like \Thm{dirmarg} suggest relations between flows with edge capacities,
and flow with vertex capacities. We were motivated to see if the KMS theorem (\Thm{dirtal})
could be proven from a flow perspective. 

\begin{conjecture} \label{conj:rout} Let source set $S$ have separation distance $r$.
Consider a flow network with unit edge capacities and vertex capacities $r^2$. The 
maxflow is at least $\Omega(r\dirvol{S})$.
\end{conjecture}
Note that the above is a {\em simultaneous} strengthening of \Thm{cspoin} and \Thm{cslr}; if we remove either the edge capacity restriction
or the vertex capacity restriction, then we get the above theorems. We show that \Conj{rout} implies the robust Talagrand isoperimetry theorem.

\begin{claim} \label{clm:link} \Conj{rout} implies \Thm{dirtal}.
\end{claim}

\begin{proof} Consider a Boolean function $f:\{0,1\}^d \to \{0,1\}$.
Consider any bicoloring $\chi$ of the violated edges.
Our aim is to lower bound $\sum_x \sqrt{\dirinfc{f}{\chi}{x}}$. 

Let $S$ be the set of $1$-valued points. By \Conj{rout} and the maxflow-mincut theorem, the mincut
of the flow network (where edges have capacity $1$ and vertices have capacity $r^2$) is at least $Cr \dirvol{S}$ for some constant $C > 0$. Note that all edges from $S$
to $\overline{S}$ must be cut; moreover any separation of (the endpoints of) these edges is
a valid $(S,\overline{S})$ cut. In terms of $f$,
these are precisely the violated edges. 

Let us use $\chi$ to construct a cut. For convenience, let $d(x)$ denote $\dirinfc{f}{\chi}{x}$.
If $d(x) \leq r^2$, we cut all violated edges incident to $x$. Otherwise, we cut the vertex $x$.
The total cut value is $\sum_{x: d(x) \leq r^2} d(x) + r^2 |\{x \ | \ d(x) > r^2\}|$.
By \Conj{rout}, the cut value is at least $Cr \dirvol{S}$. We split into two cases.

{\em Case 1, $\sum_{x: d(x) \leq r^2} d(x) \geq Cr\dirvol{S}/2$.}
Observe that $\sum_{x: d(x) \leq r^2} \sqrt{d(x)\cdot d(x)}  \leq r \sum_{x: d(x) \leq r^2} \sqrt{d(x)}$. 
Thus, $\sum_x \sqrt{\dirinfc{f}{\chi}{x}} \geq C\dirvol{S}/2 = C	\eps_f 2^d/2$.

{\em Case 2, $\sum_{x: d(x) \leq r^2} d(x) <C r\dirvol{S}/2$.} So $r^2 |\{x \ | \ d(x) > r^2\}| \geq C r \dirvol{S}/2$,
implying $|\{x \ | \ d(x) > r^2\}| \geq C\dirvol{S}/(2r)$. We can lower bound
$\sum_x \sqrt{d(x)} \geq r \sum_{x: d(x) > r^2} \geq C\dirvol{S}/2= C	\eps_f 2^d/2$.
\end{proof}

We believe that \Conj{rout} is stronger than \Thm{dirtal}, because it explicitly involves
the separation distance of $S$. 

As an aside, the connection between flows and directed isoperimetry
resolves an open question in~\cite{open14} (Pg 19, "Routing on the hypercube").
It is actually a direct consequence of \Thm{dirpoin}.

\begin{theorem} \label{thm:sachdeva} Let $(S, T, \phi)$ be a matched pair
where $S$ and $T$ are disjoint. There exist $|S|$ monotone edge disjoint paths
from $S$ to $T$.
\end{theorem}

\begin{proof} Consider the directed hypercube and take a mincut separating $S$ from $T$. 
Construct a Boolean function that assigns $1$ to the "$S$-side", and $0$ to the "$T$-side". 
All remaining vertices can be assigned values such that they do not participate in any monotonicity violation. 
Note that these vertices cannot be on a directed path from $S$ to $T$.
(Process vertices according to the partial order. For $x$, set $f(x)$ to be $\max_{y \prec x: f(y) \ \textrm{assigned}} f(y)$. 
If no $f(y)$ is assigned, set $f(x) = 0$. Observe that if $f(x)$ is assigned value $1$, then $x$
must be greater than some point in $S$. This means that $x$ cannot be less than any point in $T$,
and hence does not create monotonicity violations.)

This function has distance to monotonicity at least $|S|/2^n$. So by \Thm{dirpoin}, there are at least 
$|S|$ edges which have value $(1,0)$. These are precisely cut edges, from the $S$-side to the $T$-side. 
Hence, the cut value is at least $|S|$. Set up a flow problem on the directed hypercube
where every edge has unit capacity, vertices in $S$ are sources, and vertices in $T$ are sinks.
By the maxflow-mincut theorem, there is a flow of value at least $|S|$. 
This flow gives edge-disjoint paths from S to T.
\end{proof}

\subsection{Connections between conjectures} \label{sec:conj}

From the perspective of monotonicity testing and directed isoperimetry, \Conj{rout} is more important.
From a purely combinatorial (and maybe aesthetic) viewpoint, \Conj{glr} is more appealing.
We believe that a proof of \Conj{glr} will shed light on \Conj{rout}. This section is speculative,
but gives some of the original motivations for studying \Conj{glr}. 

An uncrossing argument of~\cite{ChSe13-j} relates general matched pairs to matched pairs contained in level sets.
(These arguments are in Section 2.4 of~\cite{ChSe13-j}, especially Claim 2.7.2 and Claim 2.7.3.)

\begin{lemma} \label{lem:uncross} Consider a set $S$ with separation distance $r$. There exist
a collection of matched pairs $(S_1, T_1, \phi_1), (S_2, T_2, \phi_2), \ldots$
with the following properties.
\begin{asparaitem}
	\item $\bigcup_i S_i \subseteq S$, $\bigcup_i T_i \subseteq \overline{S}$.
	\item $\sum_i |S_i| \geq \dirvol{S}/4$.
	\item Each $S_i$ (and $T_i$) is contained in a level set.
	\item No vertex is present in more than $2r$ cover graphs $G_{S_i, T_i}$.
\end{asparaitem}
\end{lemma}

The main upshot of this lemma is that one can "break up" the $(S, \overline{S})$
routing problem into a collection of $(S_i, T_i)$ routing problems, 
where the $S_i, T_i$ are level subsets. Moreover, the interaction between
the various cover graphs is limited, because of the last bullet point.	
Given that the separation distance of $S$ is $r$, we believe that for a constant
fraction (by total size) of the matched pairs $(S_i, T_i; \phi_i)$, the distance
of these pairs is $\Omega(r)$. If \Conj{glr} is true, we can route $\Omega(r|S_i|)$
units of edge disjoint flow with vertex congestion $r$. Any vertex participates
is at most $2r$ such flow. We had hoped to overlay
such flows and get an overall vertex congestion of $O(r^2)$. Unfortunately, a direct
overlay of flows leads to an edge congestion of $O(r)$, which is not useful for \Conj{rout}.
Nonetheless, it felt that a proof of \Conj{glr} with the proof techniques of \Thm{cspoin}
might yield insight into \Conj{rout}.

\bibliographystyle{alpha}
\bibliography{monotonicity,monotonicity-full}

\newcommand{\etalchar}[1]{$^{#1}$}
\begin{thebibliography}{BKKM23}

\bibitem[AC06]{AC04}
Nir Ailon and Bernard Chazelle.
\newblock Information theory in property testing and monotonicity testing in
  higher dimension.
\newblock {\em Information and Computation}, 204(11):1704--1717, 2006.

\bibitem[ACCL07]{ACCL04}
Nir Ailon, Bernard Chazelle, Seshadhri Comandur, and Ding Liu.
\newblock Estimating the distance to a monotone function.
\newblock {\em Random Structures Algorithms}, 31(3):371--383, 2007.
\newblock Prelim. version in Proc., RANDOM 2004.

\bibitem[BB21]{BeBl16}
Aleksandrs Belovs and Eric Blais.
\newblock A polynomial lower bound for testing monotonicity.
\newblock {\em SIAM Journal on Computing (SICOMP)}, 50(3):406--433, 2021.
\newblock Prelim. version in Proc., STOC 2016.

\bibitem[BBM12]{BBM11}
Eric Blais, Joshua Brody, and Kevin Matulef.
\newblock Property testing lower bounds via communication complexity.
\newblock {\em Computational Complexity}, 21(2):311--358, 2012.
\newblock Prelim. version in Proc., CCC 2011.

\bibitem[BCS18]{BlackCS18}
Hadley Black, Deeparnab Chakrabarty, and C.~Seshadhri.
\newblock A $o(d)\cdot\polylog(n)$ monotonicity tester for {B}oolean functions
  over the hypergrid $[n]^d$.
\newblock In {\em Proceedings, ACM-SIAM Symposium on Discrete Algorithms
  (SODA)}, 2018.

\bibitem[BCS20]{BlackCS20}
Hadley Black, Deeparnab Chakrabarty, and C.~Seshadhri.
\newblock Domain reduction: A $o(d)$ tester for boolean functions in
  $d$-dimensions.
\newblock In {\em Proceedings, ACM-SIAM Symposium on Discrete Algorithms
  (SODA)}, 2020.

\bibitem[BCS23a]{BlChSe23-2}
Hadley Black, Deeparnab Chakrabarty, and C.~Seshadhri.
\newblock A $d^{1/2+o(1)}$ monotonicity tester for boolean functions on
  d-dimensional hypergrids.
\newblock In {\em Proceedings, IEEE Symposium on Foundations of Computer
  Science (FOCS)}, 2023.

\bibitem[BCS23b]{BlChSe23}
Hadley Black, Deeparnab Chakrabarty, and C.~Seshadhri.
\newblock Directed isoperimetric theorems for boolean functions on the
  hypergrid and an $\widetilde{O}(n\sqrt{d})$ monotonicity tester.
\newblock In {\em Proceedings, ACM Symposium on Theory of Computing (STOC)},
  2023.

\bibitem[BCSM12]{BCG+10}
Jop Bri\"{e}t, Sourav Chakraborty, David~Garc\'{i}a Soriano, and Ari Matsliah.
\newblock Monotonicity testing and shortest-path routing on the cube.
\newblock {\em Combinatorica}, 32(1):35--53, 2012.

\bibitem[BGJ{\etalchar{+}}12]{BGJ+12}
Arnab Bhattacharyya, Elena Grigorescu, Madhav Jha, Kyoming Jung, Sofya
  Raskhodnikova, and David Woodruff.
\newblock Lower bounds for local monotonicity reconstruction from
  transitive-closure spanners.
\newblock {\em SIAM Journal on Discrete Mathematics (SIDMA)}, 26(2):618--646,
  2012.
\newblock Prelim. version in Proc., RANDOM 2010.

\bibitem[BGJJ24]{BaGr+24}
Paul Bastide, Carla Groenland, Hugo Jacob, and Tom Johnston.
\newblock Exact antichain saturation numbers via a generalisation of a result
  of lehman-ron.
\newblock {\em Combinatorial Theory}, 4(1), 2024.

\bibitem[Bha08]{Bha08}
Arnab Bhattacharyya.
\newblock A note on the distance to monotonicity of boolean functions.
\newblock Technical Report 012, Electronic Colloquium on Computational
  Complexity (ECCC), 2008.

\bibitem[BKKM23]{BrKh+23}
Mark Braverman, Subhash Khot, Guy Kindler, and Dor Minzer.
\newblock Improved monotonicity testers via hypercube embeddings.
\newblock In {\em Innovations in Theoretical Computer Science (ITCS)}, pages
  25:1--25:24, 2023.

\bibitem[BKR20]{BKR20}
Hadley Black, Iden Kalemaj, and Sofya Raskhodnikova.
\newblock Isoperimetric inequalities for real-valued functions with
  applications to monotonicity testing.
\newblock {\em arXiv}, abs/2011.09441, 2020.

\bibitem[BRY14a]{BeRaYa14}
Piotr Berman, Sofya Raskhodnikova, and Grigory Yaroslavtsev.
\newblock ${L}_p$-testing.
\newblock In {\em Proceedings, ACM Symposium on Theory of Computing (STOC)},
  2014.

\bibitem[BRY14b]{BlRY14}
Eric Blais, Sofya Raskhodnikova, and Grigory Yaroslavtsev.
\newblock Lower bounds for testing properties of functions over hypergrid
  domains.
\newblock In {\em Proceedings, IEEE Conference on Computational Complexity
  (CCC)}, 2014.

\bibitem[CCPS98]{Bills-book}
William Cook, William Cunningham, William Pulleybank, and Alexander Schrijver.
\newblock {\em Combinatorial Optimization}.
\newblock Wiley Interscience Series, 1998.

\bibitem[CDJS17]{ChDi+15}
Deeparnab Chakrabarty, Kashyap Dixit, Madhav Jha, and C~Seshadhri.
\newblock Property testing on product distributions: Optimal testers for
  bounded derivative properties.
\newblock {\em ACM Trans. on Algorithms (TALG)}, 13(2):1--30, 2017.
\newblock Prelim. version in Proc., SODA 2015.

\bibitem[CDST15]{ChenDST15}
Xi~Chen, Anindya De, Rocco~A. Servedio, and Li-Yang Tan.
\newblock Boolean function monotonicity testing requires (almost)
  ${O}(n^{1/2})$ non-adaptive queries.
\newblock In {\em Proceedings, ACM Symposium on Theory of Computing (STOC)},
  2015.

\bibitem[CS13]{ChSe13}
Deeparnab Chakrabarty and C.~Seshadhri.
\newblock Optimal bounds for monotonicity and {L}ipschitz testing over
  hypercubes and hypergrids.
\newblock In {\em Proceedings, ACM Symposium on Theory of Computing (STOC)},
  2013.

\bibitem[CS14]{ChSe13-j}
Deeparnab Chakrabarty and C.~Seshadhri.
\newblock An $o(n)$ monotonicity tester for {B}oolean functions over the
  hypercube.
\newblock {\em SIAM Journal on Computing (SICOMP)}, 45(2):461--472, 2014.
\newblock Prelim. version in Proc., STOC 2013.

\bibitem[CST14]{ChenST14}
Xi~Chen, Rocco~A. Servedio, and Li-Yang. Tan.
\newblock New algorithms and lower bounds for monotonicity testing.
\newblock In {\em Proceedings, IEEE Symposium on Foundations of Computer
  Science (FOCS)}, 2014.

\bibitem[CWX17]{Chen17}
Xi~Chen, Erik Waingarten, and Jinyu Xie.
\newblock Beyond {T}alagrand: New lower bounds for testing monotonicity and
  unateness.
\newblock In {\em Proceedings, ACM Symposium on Theory of Computing (STOC)},
  2017.

\bibitem[DGL{\etalchar{+}}99]{DGLRRS99}
Yevgeny Dodis, Oded Goldreich, Eric Lehman, Sofya Raskhodnikova, Dana Ron, and
  Alex Samorodnitsky.
\newblock Improved testing algorithms for monotonicity.
\newblock {\em Proceedings, International Workshop on Randomization and
  Computation (RANDOM)}, 1999.

\bibitem[EKK{\etalchar{+}}00]{EKK+00}
Funda Ergun, Sampath Kannan, Ravi Kumar, Ronitt Rubinfeld, and Mahesh
  Viswanathan.
\newblock Spot-checkers.
\newblock {\em J.\ Comput.\ System Sci.}, 60(3):717--751, 2000.
\newblock Prelim. version in Proc., STOC 1998.

\bibitem[FHH{\etalchar{+}}14]{open14}
Yuval Filmus, Hamed Hatami, Steven Heilman, Elchanan Mossel, Ryan O’Donnell,
  Sushant Sachdeva, Andrew Wan, and Karl Wimmer.
\newblock Real analysis in computer science: A collection of open problems.
\newblock
  \url{https://simons.berkeley.edu/sites/default/files/openprobsmerged.pdf},
  2014.

\bibitem[Fis04]{E04}
Eldar Fischer.
\newblock On the strength of comparisons in property testing.
\newblock {\em Information and Computation}, 189(1):107--116, 2004.

\bibitem[FLN{\etalchar{+}}02]{FLNRRS02}
Eldar Fischer, Eric Lehman, Ilan Newman, Sofya Raskhodnikova, and Ronitt
  Rubinfeld.
\newblock Monotonicity testing over general poset domains.
\newblock {\em Proceedings, ACM Symposium on Theory of Computing (STOC)}, 2002.

\bibitem[FR10]{FR}
Shahar Fattal and Dana Ron.
\newblock Approximating the distance to monotonicity in high dimensions.
\newblock {\em ACM Trans. on Algorithms (TALG)}, 6(3), 2010.

\bibitem[GGL{\etalchar{+}}00]{GGLRS00}
Oded Goldreich, Shafi Goldwasser, Eric Lehman, Dana Ron, and Alex Samordinsky.
\newblock Testing monotonicity.
\newblock {\em Combinatorica}, 20:301--337, 2000.
\newblock Prelim. version in Proc., FOCS 1998, with authors Goldreich,
  Goldwasser, Lehman, and Ron.

\bibitem[GGR]{GGR97}
O.~Goldreich, S.~Goldwasser, and S.~Ron.
\newblock A note of testing monotonicity.
\newblock Technical report.

\bibitem[HK03]{HK03}
Shirley Halevy and Eyal Kushilevitz.
\newblock Distribution-free property testing.
\newblock {\em Proceedings, International Workshop on Randomization and
  Computation (RANDOM)}, 2003.

\bibitem[HK08]{HK04}
Shirley Halevy and Eyal Kushilevitz.
\newblock Testing monotonicity over graph products.
\newblock {\em Random Structures Algorithms}, 33(1):44--67, 2008.
\newblock Prelim. version in Proc., ICALP 2004.

\bibitem[HY22]{HY22}
Nathaniel Harms and Yuichi Yoshida.
\newblock Downsampling for testing and learning in product distributions.
\newblock In {\em Proceedings, International Colloquium on Automata, Languages
  and Programming (ICALP)}, volume 229, pages 71:1--71:19, 2022.

\bibitem[KMS18]{KMS15}
Subhash Khot, Dor Minzer, and Muli Safra.
\newblock On monotonicity testing and boolean isoperimetric-type theorems.
\newblock {\em SIAM Journal on Computing}, 47(6):2238--2276, 2018.
\newblock Prelim. version in Proc., FOCS 2015.

\bibitem[LR01]{LR01}
Eric Lehman and Dana Ron.
\newblock On disjoint chains of subsets.
\newblock {\em Journal of Combinatorial Theory, Series A}, 94(2):399--404,
  2001.

\bibitem[PRW22]{PRW22}
Ramesh Krishnan~S. Pallavoor, Sofya Raskhodnikova, and Erik Waingarten.
\newblock Approximating the distance to monotonicity of boolean functions.
\newblock {\em Random Structures Algorithms}, 60(2):233--260, 2022.
\newblock Prelim. version in Proc., SODA 2020.

\bibitem[Ras99]{Ras99}
Sofya Raskhodnikova.
\newblock Monotonicity testing.
\newblock {\em Masters Thesis, MIT}, 1999.

\bibitem[RRS{\etalchar{+}}12]{RRSW11}
Dana Ron, Ronitt Rubinfeld, Muli Safra, Alex Samorodnitsky, and Omri Weinstein.
\newblock Approximating the influence of monotone boolean functions in
  {$O(\sqrt{n})$} query complexity.
\newblock {\em {ACM} Trans. Comput. Theory}, 4(4):11:1--11:12, 2012.
\newblock Prelim. version in Proc., RANDOM 2011.

\bibitem[SS08]{SS08}
Michael~E. Saks and C.~Seshadhri.
\newblock Parallel monotonicity reconstruction.
\newblock In {\em Proceedings, ACM-SIAM Symposium on Discrete Algorithms
  (SODA)}, 2008.

\end{thebibliography}

\end{document}